\newcommand{\bfR}{{\bf R}}
\newcommand{\bfE}{{\bf E}}

\newcommand{\bfN}{{\bf N}}

\newcommand{\ep}{\varepsilon}
\newcommand{\nn}{\nonumber}

\newcommand{\LOG}[1]{ \log \left( #1 \right)}

\newcommand{\SCR}[1]{{\mathscr #1}}

\newcommand{\CAL}[1]{{\cal #1}}

\newcommand{\MAT}[1]{\left(\begin{array}{cccccccccc}#1\end{array}\right)}

\newcommand{\D}[1]{{\mathscr D}( #1 )}

\newcommand{\COS}[1]{\cos{(#1)}}
\newcommand{\SIN}[1]{\sin{(#1)}}

\documentclass[preprint,11pt]{article}

\usepackage[top=30truemm, bottom=30truemm, left=25truemm, right=25truemm]{geometry}

\usepackage{fancyhdr}
\usepackage{amsthm}
\usepackage{amsmath,amssymb,latexsym,amsfonts,mathrsfs}

\theoremstyle{definition}
\newtheorem{theorem}{{\bf Theorem}}[section]

\newtheorem{lemma}[theorem]{{\bf Lemma}}
\newtheorem{proposition}[theorem]{{\bf Proposition}}
\newtheorem{corollary}[theorem]{{\bf Corollary}}

\newtheorem{remark}[theorem]{{\bf Remark}}

\newcounter{Exami}

\begin{document}
%%%%%%    references    %%%%%%\begin{document}

%\begin{frontmatter}


\begin{flushleft}
{ \Large \bf  Klein-Gordon equations with homogeneous time-dependent
electric fields}
\end{flushleft}
\begin{flushleft}
{\large MASAKI KAWAMOTO}\\
{  Department of Mathematics, Faculty of
 Science, Tokyo University of Science, \\ Kagurazaka, Shinjuku-ku, Tokyo
 162-8601, Japan \\ E-mail:mkawa@rs.tus.ac.jp}
\end{flushleft}
\begin{abstract}
 We consider a system associated to Klein-Gordon equations with
 homogeneous time-dependent electric fields. The upper and lower boundaries of
 a time-evolution propagator for this system were proven by Veseli\'{c}
 in 1991 for electric fields that are independent of time. We extend this result to time-dependent electric fields.
\end{abstract}

\begin{flushleft}
{\em Keywards:} Klein-Gordon Equation, \  Time-Dependent Electric Fields, \
Non-Selfadjoint Operators. \\
%% MSC[81U05]

\end{flushleft}

%\end{frontmatter}



\section{Introduction}

We investigate the dynamics of a relativistic charged particle with charge $q \neq
0$ that moves on $\bfR^n$, $n \in \bfN$, and is influenced by homogeneous time-dependent electric fields $\bfE(t) = E(t) = (E_1(t), ... , E_n(t)
)$, which satisfy $E_j(t) \in C^{1} (\bfR )$ for all $j \in
\{ 1,...,n\}$ and
\begin{align} \label{E1}
 \sup _{t \in \bfR} 
\sum_{j=0}^n|E_j ^{(k)}(t)| < E_{0,k}, \quad E_j^{(k)} (t) =
\frac{d^k E_j (t)}{dt^k}, \quad k \in  \{0,1\}, \quad j \in \{ 1,...,n\}, 
\end{align}
 where $0< E_{0,k} < \infty $ is a constant. The wave functions under
 consideration satisfy the following Klein-Gordon
equations:   
\begin{align}\label{1} 
&\begin{cases}
&(i \partial_t + q_E )^2 \psi_0 (t,x) = L(0,p) \psi_0 (t,x), \\ 
& \psi_0(0,x) = \psi_{0,0}, \quad \left\{ 
(i \partial _t + q_E )(\psi_0(t,x))
\right\} |_{t=0} = \psi_{0,1},
\end{cases}  \\ & L(0,p) = c^2p^2 +
 (mc^2)^2, \quad q_E=q_E(t,x) = qE(t) \cdot x,
\nn
\end{align}
where $x =(x_1,...,x_n)\in \bfR
^n$, $p =(p_1,...,p_n) = -i (\partial{x_1},....,\partial{x_n})  = -i \nabla$, $m >0$, and $q \in \bfR
\backslash \{0\}$ are the
position, momentum, mass, and charge of the charged particle, respectively. We let
$c>0$ denote the speed of light; the inner product of $a,\, b \in \bfR
^n$ is denoted by $a \cdot b$.  
To introduce
the main theorem, we consider the system of Veseli\'{c} \cite{Ve}. 

Let $\psi_0(t,x)$, $\psi _{0,0}$, and $\psi_{0,1}$ be equivalent to those
 in \eqref{1}. The substitutions $\psi_{0,1} (t,x) = (i \partial _t +q_E)
 \psi_{0} (t,x)$ and 
\begin{align*}
\Psi_0(t,x) := \MAT{\psi_0(t,x) \\ \psi_{0,1} (t,x)} , \quad \Psi_{0} =
 \Psi_{0} (0,x) = \MAT{\psi_{0,0 } \\ \psi_{0,1}}
\end{align*}
yield the following (Hamilton) system: 
\begin{align} \label{M4}
i \frac{\partial}{\partial t} \Psi _0 (t,x) = A_0(t) \Psi_0(t,x), \quad 
A_0(t) := \MAT{-q_E & 1 \\ L(0,p) & -q_E}, \quad \Psi_0(0,x) = \Psi_0.
\end{align}
By substituting $\SCR{K}_{\alpha}$, defined in \eqref{Q0} (see also \cite{Ve},
 (1.3)), and by using the same scheme
as that found in \cite{Ve}, we arrive at the following system on
$\SCR{H} = L^2(\bfR ^n) \times L^2(\bfR ^n)$: 
\begin{align}\label{Q1}
i \frac{\partial}{\partial t} \Phi_{0, \alpha}(t,x) = H_{0, \alpha}(t)
 \Phi_{0, \alpha} (t,x), \quad 
\Phi_{0,\alpha}(0,x)= \Phi_{0, \alpha} \in \SCR{H},
\end{align}  
where $H_{0,\alpha}(t) = \hat{H}_0(t) + P_{0, \alpha} (t)$ with 
\begin{align} \label{Q6}
\hat{H}_{0}(t) &= \MAT{-q_E & (L(0,p))^{1/2} \\ (L(0,p))^{1/2} & -q_E} ,
 \\ 
P_{0, \alpha}(t) &= \frac{ic^2}{2} \MAT{(1-2 \alpha)qE(t) \cdot p (L(0,p))^{-1} & 0
 \\ 0 & - (1+2 \alpha)qE(t) \cdot p (L(0,p))^{-1}}, \label{Q7}
\end{align} 
where $i = \sqrt{-1}$ and 
\begin{align*}
\Phi_{0, \alpha}(t,x) = \MAT{(L(0,p))^{1/4- \alpha /2} \psi_0(t,x) \\
 (L(0,p))^{-1/4 - \alpha /2} \psi_{0,1} (t,x)}.
\end{align*}
 The construction scheme of $H_{0, \alpha} (t)$ can be found in Appendix
 A or \cite{Ve}. Here, we call $U_{0, \alpha} (t)$ the propagator for $H_{0, \alpha}
(t)$ if $U_{0, \alpha} (t)$ satisfies the following equations:
\begin{align*}
& i \partial _t U_{0, \alpha} (t) = H_{0,\alpha}(t) U_{0, \alpha}(t) ,\quad  i \partial _t
 (U_{0, \alpha} (t))^{-1} = -(U_{0, \alpha}(t))^{-1} H_{0,\alpha}(t), \\  
& U_{0, \alpha} (t) (U_{0, \alpha} (t))^{-1} = U_{0, \alpha} (0) = \mathrm{Id}_{\SCR{H}}. 
\end{align*}
The solution of
\eqref{Q1} is denoted by $\Phi_{0, \alpha} (t,x) = U_{0, \alpha}(t) \Phi_{0, \alpha}$.
The main theorem of this paper proves that $ 0< \|U_{0, 0} (t)
\|_{\SCR{B}(\SCR{H})} < \infty$ as $t \to \infty$ and that for any $\alpha \neq 0$ and $\varphi \in C_0^{\infty} ({\bf R}^n)$, 
$\left\| U_{0, \alpha} (t)  \varphi (p)\right\|_{\SCR{B}(\SCR{H})} \to 0$ or $\infty$
as $t \to \infty $, where $\SCR{B} (\SCR{H})$ is the operator
norm on $\SCR{H}$. First, we analyze the asymptotic behavior of
$U_{0, \alpha}(t)$ in $t$. Unfortunately, $U_{0, \alpha} (t)$ is difficult to control for general electric fields satisfying only \eqref{E1}. Hence,
 we impose the following additional condition {\bf (E1)} on electric fields: \\ \\ 
{\bf (E1):} Let $E(t)$ satisfy \eqref{E1}, and define $b(t) = \int_0^t
qE(s) ds$. Then $b(t)$ satisfies $\lim_{t \to \infty} |b(t)| \to \infty$.
Moreover, for any vector $a \in {\bf R}^n$, there exist constants $e_0$
and $e_1$, independent of $t$ and $a$, such that 
\begin{align}\label{14}
\int_{|a+ b(s)| \leq 2E_{0,0}/(mc^2)} |b'(s)|   ds
  \leq e_0, \quad 
\int_{0}^t \frac{|b'(s) |^2 + |b''(s)|}{c^2 (a + b(s))^2 + (mc^2)^2} ds
 \leq e_1 
\end{align} 
holds. 

Models of electric fields satisfying Assumption {\bf (E1)} and 
remarks regarding this assumption can be found in Appendix B.

We define the Fourier transform $\SCR{F}_1^{+1}$ and inverse Fourier transform
$\SCR{F}_1^{-1} $ on $L^2(\bfR ^n)$ as follows:
\begin{align} \label{13}
\begin{array}{lll}
\hat{\phi}(q) &:= \SCR{F}_1^{+1}[\phi](q), \\ 
\check{\phi}(q) &:= \SCR{F}_1^{-1} [\phi](q), 
\end{array}
\quad
\SCR{F}_1^{\pm 1} [\phi] (q) := 
(2 \pi i)^{-n/2} \int_{\bfR^n} e^{\mp i q \cdot \eta} \phi(\eta) d \eta.
\end{align}   
We now state the main theorem in this paper.
\begin{theorem} \label{T2}
Set $\alpha =0$ in \eqref{Q1} and \eqref{Q7}, and suppose Assumption {\bf
 (E1)} holds. Then for all $t \in {\bf R}$, there exist $0< \Gamma _1 < \Gamma _2$, independent of $t$, such that 
\begin{align}\label{41}
\Gamma _1 \leq \left\|
U_{0,0}(t) 
\right\|_{{\mathscr B}({\mathscr H})} \leq \Gamma _2
\end{align}
holds, where ${\mathscr B} (\mathscr H)$ is the operator norm on
 $\SCR{H}$. Conversely, for any $\alpha \neq 0$ and $\Phi_{0, \alpha}
 \in \SCR{F}_1^{-1}C_0 ({\bf R}^n) \times \SCR{F}_1^{-1}C_0({\bf R}^n)$, 
\begin{align*}
\lim_{t \to \infty} \left\| U_{0, \alpha}(t) \Phi_{0, \alpha}
 \right\|_{\SCR{H}} =
\begin{cases}
 0,  & \alpha >0, \\
\infty, & \alpha <0,
\end{cases}
\end{align*}  
holds.
\end{theorem}
Herein, we say $U_{0,0}(t)$ is stable on $\SCR{H}$ and $U_{0, \alpha}(t)$, 
$\alpha \neq 0$ is unstable on $\SCR{H}$. As a corollary to Theorem \ref{T2}, we obtain the
following inequality.

\begin{corollary}\label{Q8}
Suppose Assumption {\bf (E1)} holds. Then for all $t \in {\bf R}$, ${\psi}_{0,0} \in
 H^{1/2} ({\bf R}^n)$, and ${\psi}_{0,1}
 \in H^{-1/2} ({\bf R}^n)$, there exist $0< \gamma _1< \gamma _2$ independent
 of $t$ such that 
\begin{align*} \nn 
& \gamma _1 
\left( 
 \left\| (L(0,p))^{1/4} \psi_{0,0} \right\|_{L^2(\bfR ^n)}^2 +
 \left\| 
(L(0,p))^{-1/4} \psi_{0,1}
\right\|_{L^2(\bfR ^n)}^2
\right) 
\\  & \qquad < \left\| (L(0,p))^{1/4} \psi_0 (t,x)\right\|_{L^2(\bfR ^n)}^2 +
 \left\| 
(L(0,p))^{-1/4} (i \partial _t + q_E)\psi_0(t,x)
\right\|_{L^2(\bfR ^n)}^2 \\ & \qquad \qquad  
< \gamma _2 \left( 
 \left\| (L(0,p))^{1/4} \psi_{0,0}\right\|_{L^2(\bfR ^n)}^2 +
 \left\| 
(L(0,p))^{-1/4} \psi_{0,1}
\right\|_{L^2(\bfR ^n)}^2
\right) \nn
\end{align*} 
holds, where $\psi_0(t,x)$, $\psi_{0,0}$, and $\psi_{0,1}$ are the same
 as those in \eqref{1}.
\end{corollary}

If $q_E$ is independent of time and satisfies $$ \left\| q_E u
 \right\|_{L^2(\bfR ^n)} \leq  \delta \left\|   (L(0,p))^{1/2} u
 \right\| _{L^2(\bfR ^n)}, \quad 0< \delta <1, $$ then 
 Najman \cite{Na} showed that $A_0 (t)$ generates a uniformly bounded
 propagator on $\SCR{K}_{\alpha}$. Veseli\'{c} subsequently applied Najman's scheme
 to non-decreasing constant electric fields $qE \cdot x$ and obtained
 stability in the time-evolution operator on $\SCR{K}_{0}$ and
 instability on $\SCR{K}_{\alpha}$, where $\alpha \neq 0$. Essential to the proof is the factorization of propagator $U_{0,
 \alpha }(t) = V_{\alpha} e^{-itqE \cdot x} V_{\alpha} ^{-1} $, where
 $V_{\alpha}$ is a time-independent linear operator satisfying 
 differential equations (see \cite{Ve}). By virtue of this factorization,  $\left\|U_{0, \alpha} (t) \Phi \right\|$ can be estimated by analyzing $V_{\alpha}$
 instead of $e^{-itH_{0, \alpha}}$. We try to extend this approach to
 time-dependent electric fields. First, we form 
 another factorization of $U_{0, \alpha} (t)$ since the aforementioned $V_{\alpha}$ depends on time if the electric fields depend on 
 time (i.e., $V_{\alpha}(t) e^{ib(t) \cdot x} V_{\alpha}
 (t)^{-1}$ is different from $U_{0, \alpha} (t)$). In order to form a
 new factorization, we focus on the so-called {\em Avron-Herbst formula}. We refer to Avron-Herbst \cite{AH} and Cycon-Froese-Kirsch-Simon
 \cite{CFKS}, Theorem 7.1., which consider the study of the Schr\"{o}dinger equations with time-dependent (and constant) electric fields:
\begin{align}\label{st}
\begin{cases}
i \partial_t u(t,x) &= ( p^2 /(2m) -qE(t) \cdot x) u(t,x) =:
 H_0^{\mathrm S} (t) u(t,x) ,\\
u(0,x) &= u_0,
\end{cases}
\end{align}
where $H_0^{\mathrm S} (t)$ is the {\em Stark Hamiltonian}. For a solution $u(t,x)$ to \eqref{st}, substituting $u(t,x) = e^{ib(t) \cdot x} v(t,x) $ yields $ i \partial _t v(t,x) = (p+ b(t))^2v(t,x)/(2m)$. Thus, by letting 
$$ 
v(t,x) = e^{-ia(t)}e^{-ic(t) \cdot p}e^{-itp^2/(2m)}u_0 , \quad a(t) = \int_0^t \frac{b(s)^2}{2m} ds , \quad c(t) = \int _0^t \frac{b(s)}{m} ds,
$$ 
one obtains $u(t,x) = e^{-ia(t)} e^{ib(t) \cdot x} e^{-ic(t) \cdot p}
 e^{-it p^2/(2m)}u_0$, i.e., a propagator for $H_0^S(t)$ can be
 described by $ e^{-ia(t)} e^{ib(t) \cdot x} e^{-ic(t) \cdot p} e^{-it
 p^2/(2m)}$. This factorization of the propagator is called the
 Avron-Herbst formula. This factorization has been applied to many research areas such as quantum scattering
theory and non-linear analysis (see Adachi-Ishida
 \cite{AI}, Avron-Herbst \cite{AH}, M\o ller \cite{Mo}, and
 Carles-Nakamura \cite{CN}). We attempt to apply this scheme to 
 \eqref{1}; in this process, we analyze the differential equation $-\partial_t^2 u_0 (t,x) = ({c^2(p+b(t))^2 + (mc^2)^2}) u_0 (t,x)$. 
 To consider the asymptotic behavior of solutions to this equation, we
 use the approach of Hochstadt \cite{Ho1}. At the conclusion of this paper
 (\S{4.1} \eqref{AV}), we obtain a new factorization of the
 propagator $U_{0,\alpha} (t)$.  

Our first approach to prove Theorem \ref{T2} is to reduce 
\eqref{1} to the ordinary differential equation in \eqref{5} through
the Fourier transform. A
similar approach to the case where the potential $q_E$ is
 dependent on time but independent of $x$, was studied by B\"{o}hme-Ressig \cite{BR1}, \cite{BR2}. Time-decaying dissipative wave equations were
studied by Wirth \cite{W1}, \cite{W2}. Our approach may be applicable
to such equations and other open problems such as those discussed by
Todorova-Yordanov \cite{TY}.

\section{Definitions and notation}

In this section, we introduce definitions and notation. Let $C$ be a constant where $C>0$. For $h(\tau, \eta)$,
$\tau \in \bfR$, and $\eta \in \bfR ^n$, let $h'(\tau, \eta)$ and $h''(\tau, \eta)$ be defined by  
$$ 
h'(\tau , \eta) := h^{(1)} (\tau, \eta), \quad h''(\tau, \eta) :=
h^{(2)} (\tau , \eta), \quad h^{(l)} (\tau , \eta) := \frac{\partial^l
h}{\partial \tau ^l} (\tau, \eta), 
$$ 
for $l \in \{0,1,2\}$. Moreover, let $\SCR{H}=L^2(\bfR
^n) \times L^2(\bfR ^n)$.  For $ \Phi =  (\phi_1 ,
\phi _2)^{\mathrm{T}} $ and $\Psi = {(\psi _1,  \psi _2 )^{\mathrm{T}}}$, the norm of the Hilbert space $\SCR{H}$ is defined by
 $\left\| \Phi
\right\|_{\SCR{H}}^2 = \left\| \phi_1 \right\|_{L^2(\bfR ^n)}^2 + \left\| 
\phi _2
\right\|_{L^2(\bfR ^n)}^2$ and inner product of $\SCR{H}$ is defined by 
\begin{align*}
\left( \Phi , \Psi \right)_{\SCR{H}} :
= (\phi_1, \psi _1)_{L^2(\bfR ^n)}
 + ( \phi_2, \psi _2)_{L^2(\bfR ^n)}.
\end{align*}
Let $A$, $B$, $C$, and $D$ be linear
 operators on $L^2(\bfR ^n)$, and let 
\begin{align*}
\CAL{A} = \MAT{A & B \\ C & D}.
\end{align*}
Then for $\Phi = (\phi_1,
  \phi_2)^{\mathrm{T}}$, we define  
\begin{align*}
\CAL{A} \Phi = \MAT{A & B \\ C  & D} \MAT{\phi _1 \\ \phi_2 } := \MAT{A \phi_1 +B \phi_2
 \\ C \phi_1 + D \phi_2}, 
\end{align*}
so that $\CAL{A}$ is a linear operator on $\SCR{H}$. Furthermore, for $\Phi _1 \in \SCR{H}$ and $\Phi _2 \in \SCR{H}$, if
there exists  $\Phi_3 \in \SCR{H}$ such that  
\begin{align*}
(\CAL{A} \Phi_1 , \Phi_2)_{\SCR{H}} = (\Phi_1, \Phi_3)_{\SCR{H}}
\end{align*}
holds, then we define $\Phi _3 = \CAL{A}^{\ast} \Phi_2$; it can be
easily obtained by
\begin{align*}
\CAL{A} ^{\ast} = \MAT{A & B \\ C & D} ^{\ast} = \MAT{A^{\ast} & C^{\ast} \\ B^{\ast} &
 D^{\ast}},
\end{align*}
where $A^{\ast}$, $B^{\ast}$, $C^{\ast}$, and $D^{\ast}$ are the
adjoint operators of $A$, $B$, $C$, and $D$, respectively, on $L^2(\bfR
^n)$. 
Finally, $(A)_M \Phi$ for $\Phi \in
\SCR{H}$ means
\begin{align} \label{X10}
 (A)_M \Phi = \MAT{A & 0 \\ 0 & A} \Phi
\end{align}
 for some linear operator $A$ on $L^2(\bfR ^n)$.

\section{Estimates of solutions to \eqref{1} }

First, we define $b(t) = (b_1
(t), b_2(t), ...,b_n(t))$ as
$\int_0^t qE(s) ds$ and take 
$$ u_0 (t,x) = e^{-i \int_0^t qE(s) \cdot x ds
} \psi_0(t,x)= e^{ -ib(t) \cdot x} \psi _0(t,x). $$ Then $u_0
(t,x)$ satisfies equations
\begin{align} \label{L}
& \begin{cases}
& - \partial _t ^2 u_0 (t,x) = L(t,p) u_0 (t,x), \\ 
& u_0 (0,x) = \psi_{0,0}, \quad (i \partial _t u_0 )(0,x) = \psi
 _{0,1}, 
\end{cases} 
\\ \nn
& L(t,p) := c^2 (p + b(t))^2 + (mc^2)^2, 
\end{align}
where $e^{-ib(t) \cdot  x} (i \partial _t) e^{ib(t) \cdot x} = 
i \partial _t - qE(t) \cdot x$ and $e^{-ib(t) \cdot  x} p e^{ib(t) \cdot x} = 
p +b(t)$ hold on the test function.
By the Fourier transform in \eqref{13}, \eqref{L} is transformed into 
\begin{align} \label{z1}
(\partial _t ^2 \hat{u}_0)(t,\xi) + L(t, \xi)\hat{u}_0 (t, \xi) = 0,
 \quad 
\hat{u}_0(0,\xi) = \hat{\psi}_{0,0}, \quad 
(i\partial _t \hat{u}_0) (0, \xi) = \hat{\psi}_{0,1}.
\end{align} 
Hence, for any fixed $\xi \in {\bf R}^n$, define $\zeta _j (t, \xi)$, $j \in \{0,1\}$, as the solution to
\begin{align}
\zeta _j ''(t, \xi) + L(t,\xi) \zeta _j (t, \xi) = 0, \quad 
\begin{cases}
\zeta _0 (0, \xi)=1, \\
\zeta' _0 (0, \xi)= 0,
\end{cases}
 \quad 
\begin{cases}
\zeta _1 (0, \xi)=0, \\
\zeta' _1 (0, \xi)= 1.
\end{cases}
\label{5}
\end{align}
Note that the solutions of \eqref{1} can be written as  
\begin{align} \label{L4}
\psi _0 (t,x) = e^{ib(t) \cdot x} \SCR{F}^{-1}_1 \zeta _0 (t, \xi)
 \hat{\psi}_{0,0} + e^{ib(t)\cdot x} \SCR{F}^{-1}_1 \zeta _1 (t, \xi)
 \hat{\psi}_{0,1}. 
\end{align}

\subsection{Hochstadt type solutions}

Let $\psi_0 (t,x)$ be
a solution to the Klein-Gordon equations in \eqref{1}. Noting \eqref{L4}, it
is equivalent to analyze the asymptotic behavior of the solution to
\eqref{5} and analyze the asymptotic behavior of the solution to \eqref{1}. To analyze \eqref{5}, we consider the approach of 
Hochstadt \cite{Ho1} (also, see Hochstadt \cite{Ho2}). For simplicity, we
denote 
\begin{align}
L(t,\xi)^{1/2} = Q(t,\xi), \quad L(t,p)^{1/2} = Q(t,p) \label{Q20}
\end{align}
in the following. Suppose that $\zeta _j
(t, \xi)$ and $\zeta _j '(t, \xi)$ are represented by   
\begin{align}   \label{L5}
& \begin{cases}
&\zeta _0 (t,\xi) = A(t,\xi) \COS{B(t,\xi)} , \ \zeta_0' (t,\xi) =
 -A(t,\xi)Q(t,\xi) \SIN{B(t, \xi)},  \\
\ &A(0, \xi)=1 , \ B(0, \xi)=0, 
\end{cases}
\\ 
& \begin{cases}
&\zeta _1 (t, \xi) = C(t,\xi) \SIN{D(t, \xi)} , \ \zeta_1' (t, \xi)
 = C(t, \xi)Q(t, \xi) \COS{D(t, \xi)} 
, \\ &C(0, \xi) = Q(0, \xi)^{-1} , \ D(0, \xi)=0, 
\end{cases}
\label{L6}
\end{align}
respectively, for functions $A(t, \xi)$, $B(t, \xi)$, $C(t, \xi)$,
and $D(t, \xi)$. Considering \eqref{5}, \eqref{L5}, and \eqref{L6}, we obtain differential equations
\begin{align} \label{6}
&\begin{cases}
&A'(t,\xi) = - Q(t,\xi)^{-1}Q'(t,\xi) \sin ^2 ({B(t,\xi)}) A(t,\xi), \\
&B'(t,\xi) = Q(t,\xi) -  Q(t,\xi)^{-1}Q'(t,\xi) \SIN{B(t,\xi)}\COS{B(t,\xi)},
\end{cases} 
\\ \label{7}
&\begin{cases}
&C'(t,\xi) = - Q(t,\xi)^{-1}Q'(t,\xi) \cos ^2 ({D(t,\xi)} ) C(t,\xi), \\
&D'(t,\xi) = Q(t,\xi) +  Q(t,\xi)^{-1}Q'(t,\xi) \SIN{D(t,\xi)}\COS{D(t,\xi)}.
\end{cases} 
\end{align}
Equations \eqref{6} and \eqref{7} yield 
\begin{align}\nonumber
A(t,\xi) &= e^{- \int_0^t Q(s,\xi)^{-1}Q'(s,\xi) \sin ^2 ({B(s,\xi)}) ds  }, \\ 
C(t,\xi) &= Q(0 , \xi)^{-1} e^{- \int_0^t Q(s,\xi)^{-1}Q'(s,\xi)
 \cos ^2 ({D(s,\xi)})  ds},  \label{8}
\end{align}
and 
\begin{align} \nonumber
B(t,\xi) &= \int_0^t Q(s,\xi) -  Q(s,\xi)^{-1}Q'(s,\xi) \SIN{B(s,\xi)}\COS{B(s,\xi)} ds, \\
D(t,\xi) &= \int_0^t Q(s,\xi) +  Q(s, \xi)^{-1}Q'(s,\xi)
 \SIN{D(s,\xi)}\COS{D(s,\xi)} ds.  \label{K3}
\end{align}

\begin{lemma}
Functions $B(t) = B(t, \cdot)$ and $D(t) = D(t, \cdot)$ ($A(t, \cdot)$ and $C(t,
 \cdot)$) are in $C^{2} ({\bf R})$. Moreover, $B(t)$ and $D(t)$
 satisfying the integral equation \eqref{K3} are unique. 
\end{lemma}
\begin{proof}
It is obvious that $B(t)$ and $D(t)$ are included in $C^{2}({\bf R})$ since $b(t) \in C^2 ({\bf R})$ (i.e., $Q(t, \cdot)$ and $Q'(t,
 \cdot)/Q(t, \cdot) $ are in
 $C^1({\bf R})$). Hence, we only
 prove the uniqueness of $B(t)$ and $D(t)$. Further, we only prove the uniqueness of
 $B(t)$ since the uniqueness of $D(t)$ can be proven in the same way. 

First, we prove that for all $t$ and $\xi$, if $B_1(t,\xi)$ and $B_2(t, \xi)$
 satisfy \eqref{K3}, then $B_1 (t, \xi) = B_2 (t, \xi)$. Let $\ep_E
 < mc/ (E_{0,0})$ and $0\leq t \leq \ep_E$. Then by \eqref{K3} and 
\begin{align}\label{t2}
\left|
\frac{Q'(s, \xi)}{Q(s, \xi)} 
\right| = \left|
\frac{c^2(\xi+b(s)) \cdot b'(s)}{c^2 (\xi+b(s))^2 + (mc^2)^2}
\right| \leq \frac{c |b'(s)|}{Q(s, \xi)} \leq \frac{c|b'(s)|}{mc^2},
\end{align}
we have 
\begin{align}
\left| \nn
B_1(t, \xi) - B_2 (t, \xi)
\right| &= 
\left| 
\int_0^t \frac{Q'(s, \xi)}{2Q(s, \xi)} \left( 
\int_{B_2(s, \xi)}^{B_1 (s, \xi)} \frac{d}{d \tau} \sin (2 \tau) d \tau 
\right) ds
\right| \\  & \leq \frac{E_{0,0}}{ mc} \int_0^t \left| 
\int_{B_2(s, \xi)}^{B_1 (s, \xi)} \cos (2\tau) d \tau
\right| ds.  \label{tt6}
\end{align}
With \eqref{tt6} and since $ t \leq \ep _E <mc/E_{0,0} $, it follows that $\sup_{t \leq
 \ep_{E}, \xi \in {\bf R}} |B_1 (t, \xi) - B_2 (t, \xi)|=0$. For $t \leq
 2 \ep_E$, we have 
\begin{align*}
& \sup_{\ep_E \leq t \leq 2 \ep_E, \ \xi\in {\bf R}}|B_1(t, \xi) - B_2
 (t, \xi)| \\ & \quad \leq
 \frac{E_{0,0}}{mc} \left( 
 0 + \sup_{\ep_{E} \leq t \leq 2 \ep_{E}, \ \xi\in {\bf R}}
 \int_{\ep_E}^t |B_1(s, \xi) - B_2 (s, \xi)| ds
\right) .
\end{align*}  
This also implies $\sup_{\ep_E \leq t \leq 2\ep_E, \xi} |B_1(t,
 \xi)-B_2 (t, \xi)| =0$. By repeating the same calculation for $t \in [n
 \ep _E, (n+1)\ep_E]$ with $n \in {\bf N}$, the lemma holds.
\end{proof} 
First, we impose
$\hat{\psi}_{0,j}  \in C_0(\bfR ^n)$, $j \in \{0,1\}$, in \eqref{1} and define $\varphi_{j} \in C_0^{\infty}( \bfR ^n)$ such that 
\begin{align}\label{K7}
\varphi_0 (\xi) = 
1, \  \mbox{on the support of $\hat{\psi}_{0,0} (\xi)$} , \quad 
\varphi_1 (\xi) = 
1, \   \mbox{on the support of $\hat{\psi}_{0,1} (\xi)$}.
\end{align}
Noting \eqref{5}, \eqref{L5}, \eqref{L6},
\eqref{8}, $|Q(t_0, \xi)|^{-1} \leq C$, $|Q(0, \xi)|^{-1} \leq C$,
and the fact that $|Q(t_0, \xi)| \leq C$ holds on the support of $\varphi_j
(\xi)$, the following proposition immediately holds.
\begin{proposition} \label{Q3}
Let $\zeta _0 (t, \xi)$ and $\zeta _1 (t, \xi)$ be equal to those defined in
 \eqref{L5} and \eqref{L6}, respectively, and let $\varphi _0 (\xi)$ and
 $\varphi_1 (\xi)$ be equal to those defined in \eqref{K7}. Then for every fixed $t_0
 \in \bfR$,
\begin{align}
\sup_{\xi \in {\bf R}^n} |\zeta _j ^{(N)} (t_0, \xi) \varphi_j (\xi)| \leq C_{j,N} 
\label{K0}
\end{align}
 holds, where  $j
 \in \{0,1\}$ and $N \in \{0,1,2\}$.
\end{proposition}
By this proposition, $\zeta _j ^{(N)} (t_0,p) \varphi_j (p)$ can
be defined as a bounded operator on $L^2({\bf R}^n) $ through 
the Fourier transform since
\begin{align*}
\left\| 
\SCR{F}_1^{-1}\zeta _j ^{(N)} (t_0,\xi) \SCR{F}_1^{+1}\varphi_j(\xi) \psi_{0,j} 
\right\|_{L^2({\bf R}^n)} &=  \left\| \zeta _j ^{(N)} (t_0, \xi)\varphi _j (\xi) \SCR{F}_1^{+1} \psi_{0,j} 
\right\|_{L^2({\bf R}^n)} \\ &\leq C_{j,N} \left\| \psi_
 {0,j}\right\|_{L^2({\bf R}^n)}
\end{align*}
holds. It also follows that for any
fixed $t$, $\zeta
_j(t,p)$ can be defined on $\SCR{F}_1^{-1} C_0({\bf R}^n)$ since $\varphi_j (\xi)$ is independent of $t$ and satisfies $\zeta _j^{(N)} (t,p)
\varphi_j(p) \psi_{0,j} = \zeta _j^{(N)}(t,p) \psi_{0,j}$. The following proposition extends the domains of $\zeta _1(t,p)$ and $\zeta
_2 (t,p)$ from
$\SCR{F}_1^{-1}  C_0({\bf R}^n)$ to $H^{1/2} ({\bf R}^n)$ and $H^{-1/2}
({\bf R}^n)$, respectively.

\begin{proposition}\label{P1}
Suppose Assumption $\mathrm{ \bf (E1)}$ holds. Let $A(t,\xi)$ and $C(t,\xi)$ be equal to those defined in \eqref{8}. Then there exist $0< C_{0} < \infty$
 and $0 < C_{1} <  \infty$, independent of $t$ and $\xi$, such that   
\begin{align}\label{9}
& \frac{Q(0, \xi)^{1/2}}{Q(t, \xi)^{1/2}} e^{-C_{0}}  \leq
 \left|
A(t, \xi)
\right| \leq \frac{Q(0, \xi)^{1/2}}{Q(t, \xi)^{1/2}} e^{C_{0}} ,  \\ & \label{9-2}
 \frac{1}{Q(t, \xi)^{1/2}Q(0, \xi)^{1/2}} e^{-C_{1}}  \leq
\left| 
C(t, \xi) 
\right| \leq \frac{1}{Q(t, \xi)^{1/2}Q(0, \xi)^{1/2}} e^{C_{1}} 
\end{align}
hold. 
\end{proposition}
\begin{proof}
For simplicity, we denote $Q(t,\xi) = Q(t)$ and $B(t, \xi) = B(t)$. 
We only calculate the term $A(t ,\xi) \varphi (\xi)$; the term $C (t ,\xi)
 \varphi (\xi)$ can
be calculated in a similar manner. 

By simple calculations, it follows that 
\begin{align}\nn
& \int_0^{t} \frac{Q'(s)}{Q(s)} \sin ^2(B(s)) ds \\ & \quad = 
\frac{1}{2} \left(
\LOG{Q(t)}- \LOG{Q(0)} - \int_0^t  \frac{Q'(s)}{Q(s)} \COS{2 B(s)} ds
\right).  \label{20}
\end{align}
Hence, to prove Proposition \ref{P1}, it suffices to show that the last term of the right-hand side of the above equation
 is uniformly bounded in $t$ and $\xi$.
Noting \eqref{t2}, we have 
\begin{align} \nn
B'(s) &= Q(s) - Q'(s)Q(s)^{-1} \sin (2B(s))/2  \geq Q(s)/2 +
 \left(  Q(s)/2 - E_{0,0}/(2mc)
\right) \\  & \geq Q(s)/2 + \left( c |\xi + b(s)|  -
 E_{0,0}/(mc)\right) /2. \label{M0}
\end{align}
Next, we define  
\begin{align*}
\Omega &:=
\left\{ 
s \in [0,t] \ :\  |\xi  +b (s)| \leq 2 E_{0,0}/(mc ^2 )
\right\}.
\end{align*}
Then by Assumption {\bf (E1)} and \eqref{t2}, we obtain that 
\begin{align*}
\left|\int_{\Omega} \frac{Q'(s)}{2Q(s)} \cos (B(s)) ds \right| \leq
C \int_{\Omega} |b'(s)| ds  \leq C e_0
\end{align*}
is bounded and independent of $t$ and $\xi$. Conversely, on
 the region $[0,t] \backslash \Omega $, by \eqref{M0}, it always follows that 
\begin{align} \label{M1}
B'(s) \geq Q(s)/2; 
\end{align} 
hence, it also follows that 
\begin{align*}
& \int_{[0,t] \backslash \Omega } \frac{Q'(s)}{Q(s)} \cos (2B(s))
 ds \\ & = 
\left[ \frac{Q'(s)}{2Q(s)B'(s)} \sin (2B(s))\right]_{\partial ([0,t]
 \backslash \Omega )}
 - \frac12 \int_{[0,t] \backslash \Omega }  z_1(s) \SIN{2B(s)} ds, 
\end{align*}
where 
\begin{align*}
 z_1(s)   &= 
\frac{1}{(B'(s))^2 Q(s)^2} \left\{ Q''(s)Q(s)B'(s) -(Q'(s))^2 B'(s) -Q(s)Q'(s)B''(s)
\right\}
\\ &=   \frac{Q''(s)}{(B'(s))^2} - \frac{ 2(Q'(s))^2}{Q(s)(B'(s))^2} 
+ \frac{(Q'(s))^2 }{2Q(s)^2 B'(s)} \cos (2B(s)) .
\end{align*}
Since \eqref{M1} holds and 
\begin{align*}
|Q'(s)| \leq C |b'(s)|, \quad |Q''(s)| \leq C \left( |b''(s)|  +
 |b'(s)|^2 \right),   
\end{align*}
it follows that on $[0,t] \backslash \Omega$, 
\begin{align*}
|z_1 (s)| \leq C \left( 
|b''(s)| Q(s)^{-2} + |b'(s)|^2Q(s)^{-2}
\right), 
\end{align*}
where  $|Q(s)|^{-1} \leq C$ and $|b'(s)| = |E(s)| \leq
 E_{0,0}$. Hence, by Assumption {\bf (E1)}, 
\begin{align*}
\left| 
\int_0^t \frac{Q'(s) \cos (2B(s))}{Q(s)} ds 
 \right| &\leq C + Ce_0 + C \int_0^t \frac{|b'(s)|^2 + |b''(s)|}{c^2(\xi
 +b(s))^2 + (mc^2)^2} ds \\ &\leq C(1 +e_0 + e_1).
\end{align*} 
Therefore, the proposition holds.
\end{proof}

By analyzing $\zeta _1$ and $\zeta _2$, we arrive at the following theorem. 
\begin{theorem} \label{T0} 
Let $\psi_0(t,x)$, $\psi _{0,0}$, and $\psi_{0,1}$ be equal to those defined
 in \eqref{1}. Suppose Assumption $\mathrm{ \bf (E1)}$ holds and that $\hat{\psi}_{0,0} \in C_0 (\bfR ^n)$ and
 $\hat{\psi}_{0,1} \in C_0 (\bfR ^n)$. Then for all $\theta
 \in \bfR $, there exists $0< C_{0,\theta }< \infty$ such that 
\begin{align}\label{l1}
\left\| (L (0,p))^{\theta} {\psi} _0 (t,x) \right\|_{L^2(\bfR ^n)}
 \leq C_{0,\theta} |b(t)|^{(2\theta -1/2)}, \quad |t| \to \infty
\end{align}
holds. In particular, 
\begin{align}\label{Q11}
\left\| \psi _0 (t,x) \right\|_{L^2(\bfR ^n)} \leq C_{0,0} |b(t)|^{- 1/2},
 \quad |t| \to \infty
\end{align}
holds, where $C_{0, \theta}$ is a constant depending only on the volume of the support
 of $\hat{\psi}_{0,0}$ and $\hat{\psi}_{0,1}$.
\end{theorem}
Solutions to \eqref{1} when the electric fields are independent of time have been investigated (see
Narozhnyi and Nikishov \cite{NN}, Tanji
\cite{Ta}, and \cite{Ve}); rotating electric fields were
investigated by Eliezer, Raicher, and Zigler \cite{ERZ}. However,
time-decay estimates \eqref{l1} and \eqref{Q11} have not been
considered. 
\begin{proof}
On the support of $\hat{\psi}_{0,0}$ and $\hat{\psi}_{0,1}$, $Q(0,\xi)^{1/2}$ is
 bounded and 
$$ 
C |b(t)|^{2} \leq L(t, \xi)^{1/2} \leq C |b(t)|^{2} $$ holds for $t \gg 1$. Thus, the inequality  
\begin{align*}
\left\|
L(t,\xi)^{\theta}\zeta _j (t, \xi) \varphi _j (\xi) \hat{\psi}_{0,1} (t,
 \xi)  
\right\|_{L^2(\bfR ^n)} 
& \leq C \left\| L(t,\xi)^{ \theta - 1/4} L(0, \xi)^{1/4} \varphi_j (\xi)
 \hat{\psi}_j \right\|
\\ & \leq C|b(t)|^{(2 \theta -1/2)} \|
 \hat{\psi}_{0,1} (t, \xi) \|_{L^2(\bfR ^n)}
\end{align*}
holds from \eqref{9} and \eqref{9-2}, where $j \in \{0,1\}$. Therefore, Theorem \ref{T0}
 holds. 
\end{proof}

\section{Proof of Theorem \ref{T2}}

In this section, we prove Theorem \ref{T2}. First, we decompose $U_{0,\alpha} (t)$ by using Hochstadt type
representations \eqref{L5}, \eqref{L6}, and \eqref{K3}. Then, by using
this factorization of $U_{0, \alpha} (t)$, we prove the stability and
instability properties.

\subsection{Factorization of $U_{0, \alpha} (t)$}
Noting the
definition of $U_{0,\alpha}
(t)$ (see \eqref{11}), $U_{0, \alpha} (t)$ can be factorized by 
\begin{align} \nn 
 U_{0,\alpha} (t) &:= K_{\alpha}^{1/2} U_{A_0} (t) (K_{\alpha}^{1/2})^{-1}  \\
 \nn  &= (e^{ib(t) \cdot x})_M \MAT{L(t,p)^{1/4-\alpha /2} & 0 \\ 0 &
 L(t,p)^{-1/4 - \alpha /2}} \\ & \quad \times \nn
\MAT{\zeta _0 (t,p) & \zeta _1 (t,p) \\ i \zeta _0'(t,p) & i \zeta _1
 '(t,p)} \MAT{L(0,p)^{-1/4+ \alpha /2} & 0 \\ 0 & L(0,p)^{1/4 + \alpha /2}}
\\ \nn
&= \MAT{ e^{ib(t) \cdot x} {\mathscr F}_1^{-1}
 {\mathscr L}_{\alpha} (t, \xi) & 0 \\ 0 &  e^{ib(t) \cdot x} {\mathscr F}_1^{-1}
 {\mathscr L}_{\alpha} (t, \xi) } \\  \quad & \quad \times\MAT{ {\mathscr G}_0(t,\xi)
 \cos (B(t,\xi)) & {\mathscr G} _1(t, \xi) \sin (D(t,\xi)) \\ 
 -i {\mathscr G}_0 (t, \xi) \sin (B(t,\xi)) & i {\mathscr G}_1 (t, \xi)\cos (D(t,\xi))
} \MAT{ {\mathscr F}_1^{+1} & 0 \\ 0 & {\mathscr F}_1^{+1}}, \label{AV}
\end{align}
where ${\mathscr L}_{\alpha} (t,\xi) = L(t,\xi)^{-
\alpha /2} L(0,\xi)^{ \alpha /2} $, $$ 
{\mathscr G} _0(t, \xi) = e^{\int_0^t (Q'(s,\xi) \cos (2B(s, \xi))
/(2Q(s, \xi))) ds} 
$$ and $$ {\mathscr G}_1(t, \xi) = e^{-
\int _0^t (Q'(s, \xi) \sin (2D(s, \xi)) /(2 Q(s, \xi)) )
ds}.$$ 
This formula is a natural extension of the Avron-Herbst formula.

\subsection{Stability of $U_{0,0} (t)$ on $\SCR{H}$}
Here, we prove the first statement of Theorem \ref{T2}. Noting that ${\mathscr
F}_1^{-1} C_0({\bf
R}^n) \times {\mathscr F}_1^{-1}C_0({\bf R}^n)$ is dense on ${\mathscr H}$, every calculation is
done on $ \Phi _{0,0} \in {\mathscr F}_1^{-1}C_0({\bf R}^n) \times {\mathscr F}_1^{-1}C_0({\bf R}^n)$. 
By \eqref{AV} with $\alpha =0$, together with the fact that 
\begin{align}\label{L7}
e^{-C_0} \leq | {\mathscr G}_0(t, \xi)| \leq e^{C_0} , \quad e^{-C_1} \leq
 |{\mathscr G}_1 (t, \xi)| \leq e^{C_1},
\end{align}
holds by \eqref{9} and \eqref{9-2}, we have that there exists $C >0$
independent of $t$ and the support of $\Phi_{0,0}$ such that $\left\| 
U_{0,0}(t) \Phi_{0,0}
\right\|_{\SCR{H}} \leq C \left\| \Phi _{0,0}\right\|_{\SCR{H}}$ holds. By the density argument, we also have
$\left\| U_{0,0}(t)\right\|_{{\mathscr B}({\mathscr H})} \leq C$. Next, we prove $\left\| U_{0,0}(t) \Phi_{0,0}
\right\|_{\SCR{H}} \geq C >0$. Letting $\Phi_{0,0} = \left( \phi_0, \phi_1
\right)^{\mathrm{T}}$, we have
\begin{align}\label{L9}
\left\| 
U_{0,0} (t) \Phi_{0,0}
\right\|_{\SCR{H}}^2 &= \left\| \SCR{G} _0(t ,\xi) \hat{\phi}_0 
\right\|^2_{L^2(\bfR ^n)} + \left\| 
\SCR{G} _1 (t, \xi) \hat{\phi}_1
\right\|_{L^2(\bfR ^n)} ^2 \\ & \qquad - 2 \mathrm{Re} \left( 
\SCR{G} _0(t , \xi) \SCR{G} _1 (t, \xi)\sin (B(t, \xi)-D (t, \xi)) \hat{\phi}_0, \hat{\phi}_1
\right)_{L^2(\bfR ^n)}. \nn
\end{align}
Using the fact that
\begin{align*}
\zeta _0 (t, \xi) \zeta _1'(t, \xi) - \zeta _0 '(t, \xi) \zeta _1 (t,
 \xi) = 1,
\end{align*}
we obtain 
\begin{align}\label{L8}
\SCR{G} _0 (t , \xi)\SCR{G} _1(t, \xi) \cos (B(t, \xi)-D (t,\xi)) =1.
\end{align}
Inequalities \eqref{L7} and \eqref{L8} imply that for all
$t \in \bfR$ and $\xi \in \bfR ^n$, there exists $ 0 < \delta \leq 1$ such that
\begin{align*}
| \cos ( B(t, \xi) - D(t, \xi) ) | > \delta
\end{align*}
holds, i.e., 
$$
|\SIN{B(t, \xi) -D(t, \xi)}| < \sqrt{1- \delta ^2} 
$$ holds. Using this inequality, \eqref{L7}, \eqref{L9}, and 
\begin{align*}
\left\|U_{0,0} (t) \Phi _{0,0} \right\|_{\SCR{H}}^2 & \geq (1- \sqrt{1-
\delta ^2}) \left( \left\| \SCR{G} _0 (t ,\xi) \hat{\phi}_0
\right\|^2_{L^2({\bf R}^n)} + \left\|
\SCR{G}_1 (t, \xi) \hat{\phi}_1 \right\|^2_{L^2({\bf R}^n)} \right) \\ &
 \geq (1- \sqrt{1- \delta ^2}) \left(\min \{ e^{-2C_0}, e^{-2C_1}\} \right) \left\| \Phi_{0,0} \right\|^2_{\SCR{H}}, 
\end{align*}
we obtain Theorem \ref{T2}.
\subsection{Instability of $U_{0, \alpha}
  (t)$, $\alpha \neq 0$, on $\SCR{H}$}
We now complete the proof of Theorem \ref{T2}. By \eqref{AV}, for $\Psi _0 =(\psi_0, \psi_1)^{\mathrm{T}}\in \SCR{F}^{-1}C_0 (\bfR ^n) \times  
\SCR{F}^{-1}C_0(\bfR ^n)$, simple calculations show that 
\begin{align*}
\left\| 
U_{0, \alpha} (t) \Psi_{0}
\right\|_{\SCR{H}}^2 &= \left\|  \sigma_{0, \alpha}(t, \xi)\hat{\psi}_0 \right\|^2_{L^2(\bfR ^n)} +  \left\|  \sigma_{1, \alpha}(t, \xi)\hat{\psi}_1 
\right\|^2_{L^2(\bfR ^n)} \\ & \quad 
-2 \mathrm{Re} (\sigma_{0, \alpha } (t, \xi) \sigma _{1, \alpha } (t, \xi) \sin (B(t, \xi) -D(t, \xi))\hat{\psi}_0, \hat{\psi} _1)_{L^2(\bfR ^n)}
\end{align*}
holds, where $\sigma _{0, \alpha}$ and $\sigma _{1, \alpha}$ are defined
by 
\begin{align*}
\sigma _{0, \alpha} (t, \xi) = \SCR{G}_0(t, \xi) \SCR{L}_{\alpha} (t,
 \xi), \quad 
\sigma _{1, \alpha} (t, \xi) = \SCR{G}_1(t, \xi) \SCR{L}_{\alpha} (t,
 \xi).
\end{align*} 
In the same way as the proof of the stability of $U_{0,0}(t,0)$, we have that there exist $c_{00} >0$ and $\delta _{00} > 0$ such that 
\begin{align*}
\left\| 
U_{0, \alpha } (t) \Psi_{0}
\right\|_{\SCR{H}}^2 
\begin{cases} 
\leq  & c_{00} \left(\left\|  \sigma_{0, \alpha}(t, \xi)\hat{\psi}_0
 \right\|^2_{L^2(\bfR ^n)} +  \left\|  \sigma_{1, \alpha }(t, \xi)\hat{\psi}_1 \right\|^2_{L^2(\bfR ^n)}\right), \\ 
\geq  & \delta_{00}\left(\left\|  \sigma_{0, \alpha }(t, \xi)\hat{\psi}_0
 \right\|^2_{L^2(\bfR ^n)} +  \left\|  \sigma_{1, \alpha }(t, \xi)\hat{\psi}_1 \right\|^2_{L^2(\bfR ^n)}\right),
\end{cases}
\end{align*}
holds. On the other hand, by \eqref{9} and \eqref{9-2}, note that for $j \in \{0,1\}$, there exist
$0 < \tilde{c}_j < \tilde{C}_j$ such that  
\begin{align*}
&
\tilde{c}_j \left\| 
\SCR{L}_{\alpha}(t, \xi) \hat{\psi}_j
\right\|_{L^2({\bf R}^n)} ^2 \leq \left\| 
\sigma_{j, \alpha}(t,\xi) \hat{\psi }_j
\right\|^2_{L^2(\bfR ^n)} \leq \tilde{C}_j  \left\| 
\SCR{L}_{\alpha} (t,\xi) \hat{\psi}_j
\right\|_{L^2({\bf R}^n)} ^2
\end{align*}
holds, where $\SCR{L}_{\alpha} (t,\xi) := 
(L(t,\xi))^{- \alpha /2} (L(0, \xi))^{ \alpha /2}$. Clearly, $L(t, \xi)=
c^2(\xi + b(t))^2 + (mc^2)^2 \to \infty $ as $t \to \infty$ holds on $C_0({\bf R}^n) \times C_0({\bf R}^n)$; hence, it follows that
for $t \to \infty$,  
\begin{align*}
\left\| 
\sigma_{j, \alpha } (t, \xi) \hat{\psi}_0
\right\|^2_{L^2(\bfR ^n)} \to
\begin{cases}
0, & \mbox{ if } \alpha >0, \\ 
\infty, & \mbox{ if } \alpha <0,
\end{cases} 
\end{align*}
holds.

\appendix
\def\thesection{APPENDIX \Alph{section}~~}

\section{Klein-Gordon systems with electric fields}
In this section, we construct the (Hamilton) system equation in \eqref{1}. This construction is the same one in \cite{Ve}. Denote 
\begin{align*}
\Psi_0 (t,x) = \MAT{\psi_0(t,x) \\ \psi_{0,1}(t,x)} , \quad 
 \psi_{0,1}(t,x) := (
i \partial _t +q_E) \psi_0 (t,x), \quad \Psi_0 = \MAT{\psi_{0,0} \\ 
\psi_{0,1}}, 
\end{align*}
where $\psi _0 (t,x)$, $\psi_{0,0}$, and $\psi_{0,1}$ are the same as those
defined in \eqref{1}. Then $\Psi _0 (t,x)$ satisfies the following equations: 
\begin{align} \label{L1}
&i\frac{\partial}{\partial t} \Psi_0(t,x) = A_0(t) \Psi_0(t,x), \quad 
A_0(t) = \MAT{-q_E & 1 \\ L (0,p) & -q_E}, \quad \Psi _0 (0,x) = \Psi_0. 
\end{align}
Here, we set $\zeta _j (t, \xi)$ to be that defined in
\eqref{5} (or \eqref{L5} and \eqref{L6}). Focusing on $\zeta ''_j
(t,\xi) = - L(t,\xi) \zeta _j (t,\xi)$, $j \in \{0,1\}$, a propagator for
$A_0 (t)$, $U_{A_0} (t)$ can be described by 
\begin{align}\label{m1}
U_{A_0} (t) = (e^{ib(t) \cdot x})_M \MAT{ \zeta _0 (t,p) & \zeta _1 (t,p) \\
 i \zeta _0 '(t,p) & i \zeta _1 '(t,p)}.
\end{align} 
Indeed,
\begin{align*}
i \frac{\partial}{\partial t} U_{A_0} (t) &= 
(e^{ib(t) \cdot x})_M (i \partial _t -qE(t) \cdot x)_M (\SCR{F_1^{-1}})_M\MAT{\zeta _0(t,\xi) & \zeta
 _1 (t,\xi) \\ i \zeta _0'(t,\xi) & i \zeta _1 '(t,\xi)} (\SCR{F}_1^{+1})_M \\ &= 
(e^{ib(t) \cdot x})_M (\SCR{F}_1^{-1})_M \Bigg\{\MAT{-qE (t)\cdot x & 0 \\ 0 & -qE(t) \cdot x}  \MAT{\zeta _0(t,\xi) & \zeta
 _1 (t,\xi) \\ i \zeta _0'(t,\xi) & i \zeta _1 '(t,\xi)}  \\ & \qquad + 
\MAT{i\zeta _0'(t,\xi) & i\zeta _1'(t,\xi) \\ L(t,\xi) \zeta _0(t,\xi) &
 L(t,\xi) \zeta _1 (t,\xi)} \Bigg\}(\SCR{F}_1^{+1})_M \\ 
&= 
(e^{ib(t) \cdot x})_M \MAT{-q_E & 1 \\ L(t,p) & -q_E} \MAT{\zeta _0(t,p) & \zeta _1 (t,p) \\ i \zeta _0'(t,p) & i \zeta _1'(t,p)}
 = A_0(t) U_{A_0}(t),
\end{align*}
where $e^{ib(t) \cdot x} L(t,p) e^{-ib(t) \cdot x} = L(0,p)  $
and $(i \partial _t) e^{ib(t) \cdot x} = e^{ib(t) \cdot x} (i \partial _t
-b'(t) \cdot x)$.

Next, we define
\begin{align}\label{X12}
\SCR{F} = (\SCR{F}_{1}^{+1})_M , \quad \SCR{F}^{-1} = (\SCR{F}_1^{-1})_M,
\end{align}
and set
\begin{align}\label{L2}
K_{\alpha} (0,p)&= \MAT{(L (0,p) )^{1/2- \alpha}& 0 \\ 0 & (L(0,p))^{-1/2 -
 \alpha}}, \quad 
K_{\alpha}  (0.p)\Phi = \SCR{F}^{-1} K_{\alpha}(0 ,\xi)\SCR{F}\Phi,
\\   
\SCR{K}_{\alpha} &= L^{1/4-\alpha /2} L^2(\bfR ^n) \times
L^{-1/4-\alpha /2}L^2(\bfR ^n)
,  \label{Q0}
\end{align}
for $\alpha \in \bfR$ and $\Phi \in \D{K_{\alpha}}$, where $L^{j} L^2(\bfR ^n)$, $j\in \bfR$ is defined
as the norm space with respect to the norm 
\begin{align*}
\left\| 
u
\right\|_{L^{j}L^2(\bfR ^n)} := \left\| 
(L(0, \xi))^j  \hat{u} (\xi) 
\right\|_{L^2(\bfR ^n_{\xi})}, \quad  
 u\in \SCR{F}_1^{-1} \D{(L(0, \xi))^j}.
\end{align*} 
Furthermore, we define 
\begin{align} \label{K1}
& \left(u , v 
\right)_{{\SCR{K}_{\alpha}}} := \left(
{K}_{\alpha}(0,\xi) \SCR{F} u , \SCR{F} v 
\right)_{{\mathscr H}}, \quad \SCR{F}u, \ \SCR{F}v \in \D{K_{\alpha} (0,
\xi)}, \\ \nonumber
& 
K_{\alpha}^{1/2} :=  (K_{\alpha} (0,p))^{1/2} =
 \MAT{(L(0,p))^{1/4 - \alpha /2} & 0 \\ 0 &
 (L(0,p))^{-1/4 - \alpha /2}},  \\ \nonumber
& 
 \Phi_{0, \alpha} (t,x) = K_{\alpha}^{1/2} \Psi_{0, \alpha} (t,x), \quad
 \Phi_{0, \alpha} (0,x) = \Phi_{0,\alpha}
= K^{1/2}_{\alpha} \Psi _{0}. 
\end{align} 
It can be shown that for $u = (u_1,u_2)^{\mathrm{T}}$, 
\begin{align*}
&\left( 
u ,u 
\right)_{\SCR{K}_{\alpha}}  = \left( 
(L(0, \xi))^{1/2 - \alpha} \hat{u}_1, \hat{u}_1
\right)_{L^2(\bfR ^n)} + \left( (L(0, \xi))^{-1/2 - \alpha} \hat{u}_2 ,
 \hat{u} _2 \right)_{L^2(\bfR ^n)} \\ &= 
\left\| 
(L(0,\xi))^{1/4 - \alpha /2} \hat{u}_1
\right\|_{L^2(\bfR ^n)}^2 + \left\| 
(L(0, \xi))^{-1/4-\alpha /2} \hat{u}_2
\right\|_{L^2(\bfR ^n)}^2 = \left\| u \right\|_{\SCR{K}_{\alpha}}^2.
\end{align*}
Thus, $(\cdot , \cdot)_{\SCR{K}_{\alpha}}$ is the inner
product of $\SCR{K}_{\alpha}$. 
Moreover, notice that for
$\Psi_{0} \in \SCR{K}_{\alpha}$, $\left\| \Phi _{0,\alpha}\right\|_{\SCR{H}} ^2 = (K
_{\alpha} \Psi_{0}, \Psi_{0})_{\SCR{H}} = \left\|
\Psi_{0}\right\|_{\SCR{K}_{\alpha}}$, i.e., $\Phi_{0,\alpha} \in
\SCR{H}$. We then define the system 
\begin{align}\label{X15}
i \frac{\partial}{ \partial t} \Phi _{0, \alpha} (t,x) = H_{0, \alpha}
 (t) \Phi _{0, \alpha} (t,x), \quad \Phi _{0, \alpha}
 (0,x) = \Phi_{0,\alpha} , \quad 
H_{0, \alpha} (t) = K^{1/2}_{\alpha} A_0 (t) (K^{1/2}_{\alpha})^{-1}.
\end{align}
on the Hilbert space $\SCR{H}$. In the same way, $U_{0, \alpha}(t) $, the propagator for
$H_{0,\alpha} (t)$, can be written as
\begin{align}\label{11}
U_{0,\alpha}(t) = K^{1/2}_{\alpha} U_{A_0} (t) (K^{1/2}_{\alpha})^{-1}, \quad 
U_{0, \alpha}(t) ^{-1} = K_{\alpha}^{1/2}
U_{A_0} (t) ^{-1} (K_{\alpha}^{1/2})^{-1},
\end{align} 
and we obtain the system 
\begin{align} \label{50}
i \frac{d}{dt} U_{0, \alpha} (t) \Phi_{0, \alpha} = H_{0, \alpha}(t)
 U_{0, \alpha} (t) \Phi_{0, \alpha} , \quad \Phi_{0, \alpha} \in {\mathscr H}
\end{align}
with Hilbert space ${\mathscr H}$ and complex valued energy $H_{0,
\alpha} (t)$. Straightforward calculations show that $H_{0, \alpha}(t)$ can be written as
\begin{align*}
\MAT{(L(0,p))^{1/4- \alpha /2}(-q_E) (L(0,p))^{-1/4 + \alpha /2} & (L(0,p))^{1/2} \\ (L(0,p))^{1/2} & (L(0,p))^{-1/4 - \alpha /2}(-q_E)(L(0,p))^{1/4 + \alpha/2}  }.
\end{align*}
Noting that for an invertible smooth function $F$ and its inverse $F^{-1 }$, 
\begin{align*}
F(p)^{-1} x F(p) &= \SCR{F}_1^{-1} F(\xi)^{-1} \SCR{F}_1^{+1} x
 \SCR{F}_1^{-1} F(\xi) \SCR{F}_1^{+1} , \quad (\SCR{F}_{1}^{+1} x
 \SCR{F}_1^{-1} = i \nabla _{\xi}), \\ 
 &= 
 \SCR{F}_1^{-1} F(\xi)^{-1} (i \nabla F)(\xi) \SCR{F}_1^{+1} +  \SCR{F}_1^{-1} F(\xi)^{-1} F(\xi) \SCR{F}_1^{+1} x  \SCR{F}_1^{-1}
 \SCR{F}_1^{+1} \\ &= i F(p)^{-1}(\nabla F)(p) + x
\end{align*}
holds. Hence, $(L(0,p))^{- \theta} qE \cdot x (L(0,p))^{\theta} = qE \cdot x + 2 i c^2 \theta qE \cdot p (L(0,p))^{-1}$, and $H_{0, \alpha} (t)$ can be
decomposed into $\hat{H}_{0, \alpha} (t) =   \hat{H}_0 (t) + P_{0, \alpha} (t)$ ;
$\hat{H}_0 (t)$ and $P_{0, \alpha} (t)$ are the same as those defined in
\eqref{Q6} and \eqref{Q7}, respectively. Here, $\hat{H}_0 (t)$ is a symmetric operator
(self-adjoint operator for every fixed $t$, see Lemma 2.1. of \cite{Ve}), but $P_0(t)$ is a non-symmetric operator (clearly, it is a
complex valued operator).

\section{Models of time-dependent electric fields}
Here, we
 give examples of electric fields satisfying Assumption {\bf
 (E1)}. First, we assume that $b(t)$ satisfies $b(t) = (0, 0, ..,
 0,b_j(t),0,...,0)$, $j \in \{1,2,...,n\}$, and $b_j(t)$ can be written as 
\begin{align} \label{Mo1}
b_j(t) = 
\begin{cases}
C_{\gamma} t^{\gamma} + \rho_{\gamma} (t)  & 0<
 \gamma < 1, \\  
C_1 t + \rho_1 (t) + \theta _1 (t) & \gamma = 1, 
\end{cases}
\end{align}
where $C_{\gamma} \neq 0$ is a constant, $\rho_{\gamma} \in C^2({\bf R}^n) $
satisfies $| \rho^{(l)}_{\gamma} (t) | = o(t^{\gamma - l}) $ for $l \in
\{0,1,2\}$, and $| \theta ^{(l)}_{1} (t) | \leq C $ for $l\in\{0,1,2\}$. It
can easily be shown that 
\begin{align} \nn 
\int_{|a + b(s)| \leq 2E_{0,0}/(mc^2)} |b'(s)| ds & \leq 
\int_{|a_j +b_j(s) | \leq 2E_{0,0}/(mc^2)} |b_j'(s)| ds \\ & \leq \left|
 \int_{|\tau
 | \leq 2E_{0,0}/(mc^2)} \frac{|b_j'(s)|}{b_j'(s)} d \tau \right| \leq C
 \label{Mo2}
\end{align}
and 
\begin{align*}
&\int_0^t \frac{|b'(s)|^2 + |b''(s)|}{Q(s, a) ^2} ds \\ & \leq C_{R} + \int_{R}^t
 \frac{ |b'_j(s)|^2+ |b_j''(s)| }{c^2(a_j + C_{\gamma} s^{\gamma} +
 \rho_{\gamma} (s) + \theta _{\gamma} (s)
 )^2 + (mc^2)^2} ds  \\ & \leq  C_{R} + C \sup_{s > R} \left|  s^{1- \gamma}
 (|b'_j(s)|^2 + |b''_j(s) |) \right|
 \int _{-\infty }^{\infty}
 \frac{d \tau}{c^2 (\tau + \theta _{\gamma} (s))^2)  + (mc^2)^2 } d \tau  
\end{align*}
hold, where $\theta _{\gamma} (s) \equiv 0$ for $\gamma < 1$. By dividing the limits of integration into two regions, $|\tau| \leq
 2 |\theta _{\gamma } (s)| \leq C
 $ and $|\tau | \geq 2 |\theta _{\gamma} (s)|$,  notice that the last term
 of the above inequality is smaller
 than 
\begin{align*} 
C \sup_{s > R} \left|  s^{1- \gamma}
 ( |b'_j(s)|^2 + |b''_j(s) |  ) \right| \left( \int_{|\tau| \leq C} d \tau 
  +  \int_{|\tau| \geq 2|\theta _{\gamma} (s)|}
 \frac{d \tau}{c^2 \tau ^2 /4 + (mc^2)^2}  \right)
\leq C,
\end{align*}
where \eqref{Mo1} is utilized.  \\ 
Next, assume $b(t) = (0,...,0,b_j(t),0...,0)$ and $b_j(t)$
 can be written as 
\begin{align*}
 b_j(t) = e_3 (\log (1 + e_4 |t|)), 
\end{align*}
where $e _3 \neq 0$ and $e_4 > 0$ are constants. By the same
approach as \eqref{Mo2},we obtain the left-hand side of \eqref{14} for
this particular $b(t)$. Moreover, by using the fact that $(b'_j (s)) ^2$ and $b''_j(s)$ are
integrable on $[R, \infty)$, the right-hand side of
\eqref{14} can also be obtained for this $b(t)$.

\begin{remark}
Suppose $b(t)$ satisfies $b(t) = (0,...,0,b_{j1}(t),0,...,0,b_{j2}
 (t),0,...,0 )$ and $b_{j1} (t)$ and $b_{j2} (t)$ are written in the
 same form as \eqref{Mo1} by replacing $\gamma \to \gamma _1$ and
 $\gamma \to \gamma _2$, respectively.
 Then it is sufficient to consider the same approach as above for
 the maximum of $\{ \gamma _1 , \gamma _2\}$; indeed, suppose $\gamma
 _1 \geq \gamma _2$. Noting  that
\begin{align*}
\int_{| a + b(s)| \leq 2E_{0,0}/(mc^2)} |b'(s)| ds \leq C_R +
C \int_{ \tiny\begin{array}{ll}
|a_{j1} + b_{j1}(s)|  \leq 2E_{0,0}/(mc^2) \\   \qquad \qquad \qquad s \geq R 
\end{array}} |b_{j1}'(s)| ds 
\end{align*} 
and 
\begin{align*}
\int_R^t \frac{|b'(s)|^2 + |b''(s)|}{Q(s, a)^2} ds \leq 
C \int_{R}^{t} \frac{|b'_{j1} (s)|^2 + |b''_{j1} (s)|}{c^2(a_{j1} +
 b_{j1} (s))^2 + (mc^2)^2} ds, 
\end{align*}
it is straightforward to  prove that \eqref{14} mimics the above approach. Similarly, we consider the case when $b(t) = (b_1 (t),
 ...., b_n(t))$. However, if AC electric
 fields are included in $E(t)$, \eqref{14} is difficult to prove. For example, consider the case when $b_{j1}(t) = t^{ \gamma}$ and $b_{j2} (t) = t^{
 \gamma /2} + \cos t$ with $0< \gamma <1$, i.e., $| b_{j1}(t) | \geq
 |b_{j2} (t)|$ holds for $t \gg 1$, but $|b^{(1)}_{j1} (t)|
 \geq |b^{(l)}_{j2} (t)|$, $l \in \{1,2\}$, is not always true. Clearly,
 $s^{1- \gamma} (|b''(s)| + |b'(s)|)$ is not bounded; hence, our proof
 fails. Other approaches must be established 
 to consider more general electric fields including AC electric
 fields.
\end{remark}

\end{document}